\documentclass{article}

\usepackage{fullpage}

\usepackage{cite}

\usepackage{amsfonts}
\usepackage{amsmath}
\usepackage{amsthm}
\usepackage{url}
\usepackage{graphicx}
\usepackage{ifthen}
\usepackage{wrapfig}
\usepackage{algorithm}
\usepackage{algorithmic}

\usepackage{tikz}
\usetikzlibrary{shapes,arrows,decorations,decorations.pathmorphing,decorations.pathreplacing,plotmarks}

\usepackage{subfig}

\newcommand{\para}{\paragraph}

\newcommand{\newparentheses}[3]{%
  \expandafter\newcommand\csname #1\endcsname[1]{#2##1#3}%
  \expandafter\newcommand\csname #1L\endcsname[1]{\bigl#2##1\bigr#3}%
  \expandafter\newcommand\csname #1XL\endcsname[1]{\Bigl#2##1\Bigr#3}%
  \expandafter\newcommand\csname #1V\endcsname[1]{\left#2##1\right#3}}

\newparentheses{parens}{(}{)}
\newparentheses{braces}{\{}{\}}
\newparentheses{brackets}{[}{]}
\newparentheses{floor}{\lfloor}{\rfloor}
\newparentheses{ceil}{\lceil}{\rceil}
\newparentheses{abs}{|}{|}
\newparentheses{set}{\{}{\}}
\newparentheses{size}{|}{|}
\newparentheses{seq}{\langle}{\rangle}

\makeatletter
\newcommand{\onenewattribute}[4]{%
  \@ifundefined{#2}{\let\@@def\newcommand}{\let\@@def\renewcommand}%
  \expandafter\@@def\csname #2\endcsname[1][]{%
    \def\first@arg{##1}\csname @#2\endcsname}%
  \@ifundefined{@#2}{\let\@@def\newcommand}{\let\@@def\renewcommand}%
  \expandafter\@@def\csname @#2\endcsname[2][]{%
    \ifthenelse{\equal{#1}{sub}}%
    {\csname @@#2\endcsname{##1}{\first@arg}{##2}}%
    {\csname @@#2\endcsname{\first@arg}{##1}{##2}}}
  \@ifundefined{@@#2}{\let\@@def\newcommand}{\let\@@def\renewcommand}%
  \expandafter\@@def\csname @@#2\endcsname[3]{%
    \ifthenelse{\equal{##1}{}}%
    {\ifthenelse{\equal{##2}{}}%
      {#3\csname #4\endcsname{##3}}%
      {#3_{##2}\csname #4\endcsname{##3}}}%
    {\ifthenelse{\equal{##2}{}}%
      {#3^{##1}\csname #4\endcsname{##3}}%
      {#3_{##2}^{##1}\csname #4\endcsname{##3}}}}}
\newcommand{\newattribute}[3][sub]{%
  \onenewattribute{#1}{#2}{#3}{parens}%
  \onenewattribute{#1}{#2L}{#3}{parensL}%
  \onenewattribute{#1}{#2XL}{#3}{parensXL}%
  \onenewattribute{#1}{#2V}{#3}{parensV}}

\newcommand{\newproperty}[3][sub]{%
  \@ifundefined{#2}{\let\@@def\newcommand}{\let\@@def\renewcommand}%
  \expandafter\@@def\csname #2\endcsname[2][]{%
    \ifthenelse{\equal{#1}{sub}}%
    {\ifthenelse{\equal{##1}{}}%
      {#3_{##2}}%
      {#3_{##2}^{##1}}}%
    {\ifthenelse{\equal{##1}{}}%
      {#3^{##2}}%
      {#3_{##1}^{##2}}}}}
\makeatother

\newattribute{OhOf}{\mathcal{O}}
\newattribute{ThetaOf}{\Theta}
\newattribute{OmegaOf}{\Omega}
\newattribute{ohOf}{\mathcal{o}}
\newattribute{omegaOf}{\omega}

\newattribute{scc}{\textrm{scc}}
\newattribute{prd}{\textrm{prd}}
\newcommand{\M}{\mathcal{M}}

 \newtheorem{theorem}{Theorem}
 \newtheorem{lemma}{Lemma}
 \newtheorem{corollary}{Corollary}

\makeatletter
\def\inmod#1{\allowbreak\mkern5mu({\operator@font mod}\,\,#1)}
\makeatother

\newcommand{\C}{\mathcal{C}}
\newcommand{\s}{\mathcal{S}}
\newcommand{\E}{\mathbb{E}}
\newcommand{\hev}{\mbox{\tiny heavy}}

\title{Sorting and Permuting without Bank Conflicts on GPUs}

\author{Peyman Afshani \thanks{Work supported in part by the Danish National Research Foundation grant DNRF84 through Center for Massive Data Algorithmics (MADALGO)} \\ 
   MADALGO, Aarhus University, Denmark
  \and 
  Nodari Sitchinava \\ 
   University of Hawaii, Manoa, HI, USA
}

\date{}

\begin{document}
\maketitle
\begin{abstract}

In this paper, we look at the complexity of designing algorithms without any bank conflicts in the shared memory of Graphical Processing Units (GPUs).  Given input of size $n$, $w$ processors and $w$ memory banks, we study three fundamental problems: sorting, permuting and $w$-way partitioning (defined as sorting an input containing exactly $n/w$ copies of every integer in $[w]$).  

We solve sorting in optimal $O(\frac{n}{w} \log n)$ time.  When $n \ge w^2$, we solve the partitioning problem optimally in $O(n/w)$ time. We also present a general solution for the partitioning problem which takes $O(\frac{n}{w} \log^3_{n/w} w)$ time.  Finally, we solve the permutation problem using a randomized algorithm in $O(\frac{n}{w} \log\log\log_{n/w} n)$ time. Our results show evidence that when working with banked memory architectures, there is a separation between these problems and the permutation and partitioning problems are not as easy as simple parallel scanning.  
\end{abstract}

\bigskip

\centerline{{\bf Keywords}: GPU, model, algorithm design, efficient algorithms, prefix sums, sorting.}



\section{Introduction}

Graphics Processing Units (GPUs) over the past decade have been transformed
from special-purpose graphics rendering co-processors, to a powerful platform
for general purpose computations, with runtimes rivaling best implementations
on many-core CPUs. With high memory throughput, hundreds of physical cores and
fast context switching between thousands of threads, they became very popular among
computationally intensive applications. Instead of citing a
tiny subset of such papers, we refer the reader to {\tt gpgpu.org} website~\cite{gpgpu.org}, which lists over 300 research papers on this topic.

Yet, most of these results are experimental and the theory community seems to
shy away from designing and analyzing algorithms on GPUs. In part, this is
probably due to the lack of a simple theoretical model of computation for GPUs.
This has started to change recently, with introduction of several theoretical
models for algorithm analysis on GPUs.

\para{A brief overview of GPU architecture.}
A modern GPU contains hundreds of physical cores. To implement such a large
number of cores, a GPU is designed hierarchically. It consists of a number of
{\em streaming multiprocessors (SMs)} and a {\em global memory} shared by all
SMs. Each SM consists of a number of cores (for concreteness, let us
parameterize it by $w$) and a {\em shared memory} of limited size which is shared 
by all the cores within the SM but is inaccessible by other SMs.
With computational power of hundreds of cores, 
latency of accessing memory becomes non-negligible and GPUs take several
approaches to mitigate the problem. 
\begin{figure}[tb]
\begin{center}
  \includegraphics[width=.5\textwidth]{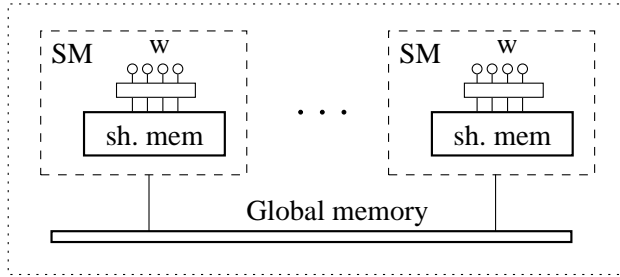}
\end{center}
\caption{A schematic of GPU architecture}
\end{figure}

First, they support massive hyper-threading, i.e., multiple logical threads may
run on each physical core with light context switching between the 
threads. Thus, when a thread stalls on a memory request, the other threads can
continue running on the same core. To schedule all these threads efficiently,
groups of $w$ threads, called {\em warps}, are scheduled to run on $w$ physical
cores simultaneously in {\em single instructions, multiple data
(SIMD)}~\cite{flynn:simd} fashion.

Second, there are limitations on how data is accessed in memory. Accesses to
global memory are most efficient if they are {\em coalesced}. Essentially, it
means that $w$ threads of a warp should access contiguous $w$ addresses of
global memory. On the other hand, shared memory is partitioned into $w$ {\em
memory banks} and each memory bank may service at most one thread of a warp in
each time step. If more than one thread of a warp requests access to the same
memory bank, a {\em bank conflict} occurs and multiple
accesses to the same memory bank are sequentialized. Thus, for optimal
utilization of processors, it is recommended to design algorithms that
perform coalesced accesses to global memory and incur no bank conflicts in
shared memory~\cite{nvidia:guide}.

\para{Designing GPU algorithms.}
Several papers~\cite{nakano:DMM,haque:model,sitchinava:gpu-model,ma-model}
present theoretical models that incorporate the concept of coalesced accesses
to global memory into the performance analysis of GPU algorithms. In essence,
all of them introduce a complexity metric that counts the number of blocks
transferred between the global and internal memory (shared memory or
registers), similar to the I/O-complexity metric of sequential and parallel
external memory and cache-oblivious models on
CPUs~\cite{av:em-model,arge:pem-model,frigo99,blelloch:multicore-dnc}. The
models vary in what other features of GPUs they incorporate and, consequently,
in the number of parameters introduced into the model.

Once the data is in shared memory, the usual approach is to implement standard
parallel algorithms in the PRAM model~\cite{jaja:book} or interconnection
networks~\cite{leighton:book}. For example, sorting data in shared memory, is
usually implemented using sorting networks, e.g. Batcher's odd-even
mergesort~\cite{batcher:networks} or bitonic mergesort~\cite{batcher:networks}.
Even though these sorting networks are not asymptotically optimal, they provide
good empirical runtimes because they exhibit small constant factors, they fit
well in the SIMD execution flow within a warp, and for small inputs asymptotic
optimality is irrelevant. On very small inputs (e.g. on $w$ items -- one per memory
bank) they also cause no bank conflicts. 

However, as the input sizes for shared memory algorithms grow, bank conflicts
start to affect the running time. algorithm design. 

The first paper that studies bank conflicts on GPUs is by Dotsenko et
al.~\cite{dotsenko:scan}. The authors view shared memory as a two-dimensional
matrix with $w$ rows, where each row represents a separate memory bank. Any
one-dimensional array $A[0..n-1]$ will be laid out in this matrix in
column-major order in $\ceil{n/w}$ contiguous columns. Note, that {\em strided}
parallel access to data, that is each thread $t$ accessing array entries
$A[wi+t]$ for integer $0 \le i < \ceil{n/w}$, does not incur bank conflicts
because each thread accesses a single row. The authors also observed that the
{\em contiguous} parallel access to data, that is each thread  scanning a
contiguous section of $\ceil{n/w}$ items of the array,  also incurs no bank
conflicts if $\ceil{n/w}$ is co-prime with $w$. Thus, with some extra padding,
contiguous parallel access to data can also be implemented without any bank
conflicts. 

Instead of adding padding to ensure that $\ceil{n/w}$ is co-prime with $w$,
another solution to bank-conflict-free contiguous parallel access is to convert
the matrix from column-major layout to row-major layout and perform strided
access. This conversion is equivalent to in-place transposition of the matrix.
Catanzaro et al.~\cite{catanzaro:transpose}  study this problem and present
an elegant bank-conflict-free parallel algorithm that runs in $\ThetaOf{n/w}$ time, which is optimal.

Sitchinava and Weichert~\cite{sitchinava:gpu-model} present a strong
correlation between bank conflicts and the runtime for some problems. Based on
the matrix view of shared memory, they developed a sorting network that incurred
no bank conflicts. They show that although compared to Batcher's sorting
networks their solution incurs extra $\ThetaOf{\log n}$ factor in parallel
time and work, it performs better in practice because it incurs no bank
conflicts.

Nakano~\cite{nakano:DMM} presents a formal definition of a parallel model with the matrix view of shared memory, which is also extended to model memory access
latency hiding via hyper-threading.\footnote{Nakano's DMM exposition actually
swapped the rows and columns and viewed memory banks as columns of the matrix
and the data laid out in row-major order.} He calls his model {\em Discrete Memory
Model (DMM)} and studies the problem of {\em offline permutation}, which we define in detail later.

The DMM model is probably the simplest abstract model that captures the important aspects of designing bank-conflict-free algorithms for GPUs. In this paper we will work in this model. However, to simplify the exposition, we will assume that each memory access incurs no latency (i.e. takes a unit time) and we have exactly $w$ processors. This simplification still captures the key algorithmic challenges of designing bank-conflict-free algorithms without the added complexity of modeling hyper-threading. We summarize the key features of the model below, and for more details refer the reader to~\cite{nakano:DMM}.

\para{Model of Computation.}
Data of size $n$ is laid out in memory as a matrix $\M$ with dimensions $w \times m$, where $m = \ceil{n/w}$. As in the PRAM model, $w$ processors proceed synchronously in discrete time steps, and in each time step perform access to data (a processor may skip accessing data in some time step).  Every processor can access any item within the matrix. However, an algorithm must ensure that in each time step at most one processor accesses data within a particular row.\footnote{This is analogous to how EREW PRAM model requires the algorithms to be designed so that in each time step at most one processor accesses any memory address.} Performing computation on a constant number of items takes unit time and, therefore, can be completed within a single time step. The parallel time complexity (or simply {\em time}) is the number of time steps required to complete the task. The work complexity (or simply {\em work}) is the product of $w$ and the time complexity of an algorithm.\footnote{Work complexity is easily computed from time complexity and number of processors, therefore, we don't mention it explicitly in our algorithms. However, we mention it here because it is a useful metric for efficiency, when compared to the runtime of the optimal sequential algorithms.}

Although the DMM model allows each processor to access any memory bank (i.e. any row of the matrix), to simplify the exposition, it is helpful to think of each processor fixed to a single row of the matrix and the transfer of information between the processors being performed  via ``message passing'', where 
at each step, processor $i$ may send a message (constant words of information) to another
processor $j$ (e.g., asking to read or write a memory location within row $j$). 
Next, the processor $j$ can respond by sending constant words of information back to row $i$.
Crucially and to avoid bank conflicts, we demand that at each parallel computation step,
at most one message is received by each row; we call this the ``Conflict Avoidance Condition'' or CAC.

Note that this view of interprocessor communication via message passing is only
done for the ease of exposition, and algorithms can be implemented in the DMM model (and
in practice) by processor $i$ directly reading or writing the contents of the
target memory location from the appropriate location in row $j$. Finally, CAC
is equivalent to each memory bank being access by at most one processor in each
access request made by a warp.

\para{Problems of interest.}
Using the above model, we study complexity of developing bank conflict free algorithms for the following fundamental problems:

$\bullet$ \textit{Sorting}: The matrix $\M$ is populated with items from a totally ordered universe. The goal
    is to have $\M$ sorted in row-major order.\footnote{ The final layout within the matrix (row-major or column-major order) is of little relevance, because the conversion between the two layouts can be implemented efficiently in time and work required to simply read the input~\cite{catanzaro:transpose}. }

$\bullet$ \textit{Partition}: The matrix $\M$ is populated with labeled items. The labels form a permutation that contains $m$ copies of every integer in $[w]$ and an item with label $i$ needs to be sent to row $i$.

$\bullet$ \textit{Permutation}: The matrix $\M$ is populated with labeled items. The labels form a permutation of tuples $[w] \times [m]$.  And an item with label $(i,j)$ needs to be sent to the $j$-th memory location in row $i$.

While the sorting problem is natural, we need to mention a few remarks regarding the permutation and the partition problems.
Often these two problems are equivalent and thus there is little motivation to separate them.
However rather surprisingly, it turns out that in our case these problems are in fact very different. 
Nonetheless, the permutation problem can be considered an ``abstract'' sorting problem where all the
comparisons have been resolved and the goal is to merely send each item to its correct location.
The partition problem is more practical and it appears in scenarios where the goal is to split an input into
many subproblems. For example, consider a quicksort algorithm with pivots, $p_0 = -\infty, p_1, \cdots, p_{w-1}, p_w = \infty$. 
In this case, row $i$ would like to send all the items that are greater than $p_{j-1}$ but less than $p_j$ to row $j$ but row $i$
would not know the position of the items in the final sorted matrix.
In other words, in this case, for each element in row $i$, we only know the destination row rather than the
destination row and the rank.

\para{Prior work.}
Sorting and permutation are among the most fundamental algorithmic problems. 
The solution to the permutation problem is trivial in the random access memory
models -- $\OhOf{n}$ work is required in both RAM and PRAM models of
computation -- while sorting is often more difficult (the $\Omega(n \log n)$ comparison-based lower bound is a
very classical result). However, this picture changes in other models.
For example, in both the sequential and parallel external memory models~\cite{av:em-model,arge:pem-model},
which model hierarchical memories of modern CPU processors, existing lower bounds
show that permutation is as hard as sorting (for most realistic parameters of the models)~\cite{av:em-model,gero:thesis}.

In the context of GPU algorithms, matrix transposition was studied as a special case of the permutation problem by
Catanzaro et al.~\cite{catanzaro:transpose} and they showed that one can
transpose a matrix in-place without bank conflicts in $\OhOf{n}$ work. While
they didn't explicitly mentioned the DMM model, their analysis holds trivially
in the DMM model with unit latency. 
Nakano~\cite{nakano:DMM} studied the problem of performing arbitrary
permutations in the DMM model {\em offline}, where
the permutation is known in advance and we are allowed to
pre-compute some information before running the algorithm.
The time to perform the precomputation is not counted toward the complexity of
performing the permutation. Offline permutation is useful if the permutation
is fixed for a particular application and we can encode the
precomputation result in the program description. Common examples of fixed
permutations include matrix transposition, bit-reversal permutations, and FFT permutations. 
Nakano showed that any offline permutation can be implemented in linear work.  
The required pre-computation in Nakano's algorithm is coloring of a regular bipartite graph, 
which seems very difficult to adapt to the {\em online} permutation problem.

As mentioned earlier, Sitchinava and Weichert~\cite{sitchinava:gpu-model}
presented the first algorithm for sorting $w \times w$ matrix which incurs no
bank conflicts. 
They use Shearsort~\cite{shearsort}, which repeatedly sorts columns of the matrix in
alternating order and rows in increasing order. 
After $\ThetaOf{\log w}$ repetitions, the matrix is sorted in column-major order. 
Rows of the matrix can be sorted without bank conflicts. 
And since matrix transposition can be
implemented without bank conflicts, the columns can also be sorted without bank
conflicts via transposition and sorting of the rows. 
The resulting runtime is $\ThetaOf{t(w)\log w}$, where $t(w)$ is the time it takes to sort an array of
$w$ items using a single thread.

\para{Our Contributions.}
In this paper we present the following results.

$\bullet$ \textit{Sorting}: We present an algorithm that runs in $O(m \log (mw))$ time, which is optimal in the context of comparison based algorithms.

$\bullet$ \textit{Partition}: We present an optimal solution that runs in $O(m)$ time when $w \le m$.
    We generalize this to a solution that runs in $O(m \log^3_m w)$ time.

$\bullet$ \textit{Permutation}: We present a randomized algorithm that runs in expected $O(m \log\log\log_m w)$ time.  Even though this is a rather technical solution (and thus of theoretical interest), it strongly hints at a possible separation between the partition and permutation problems in the DMM model.


\section{Sorting and Partitioning}
In this section we improve the sorting algorithm of Sitchinava and Weichert~\cite{sitchinava:gpu-model} by removing a  $\ThetaOf{\log w}$ factor.
We begin with our base case, a ``short and wide''  $w \times m$ matrix $\M$ where $w \le \sqrt{m}$.
The algorithm repeats the following twice and then sorts the rows in ascending order.

\begin{enumerate}
  \item Sort rows in alternating order (odd rows ascending, even rows descending).
  \item Convert the matrix from row-major to column-major layout (e.g., the first $w$ elements of the first row
      form the first column, the next $w$ elements form the second column and so on).
  \item Sort rows in ascending order.
  \item Convert the matrix from column-major to row-major layout (e.g., the first $m/w$ columns will form the first row).
\end{enumerate}

\begin{figure}[t]
    \centering
    \includegraphics{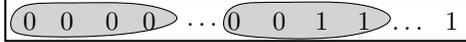}
    \caption{Consider one row that is being converted to column-major layout. The elements that will be put in the same column
    are bundled together. It is easily seen that each row will create at most one dirty column.}
    \label{fig:dirt}
\end{figure}
\begin{lemma}\label{lem:fat}
The above algorithm sorts the matrix correctly in $O(m \log m)$ time.
\end{lemma}
\begin{proof}
    We would like to use the 0-1 principle but since are not working in a sorting network, we simply
    reprove the principle. 
    Observe that our algorithm is essentially oblivious to the values in the matrix and only takes into
    account the relative order of them.
    Pick a parameter $i$, $1 \le i \le mw$, that we call the \textit{marking value}.
    Based on the above observation, we attach a mark of  ``0'' to any element that has rank less than $i$
    and ``1'' to the remaining elements. These marks are symbolic and thus invisible to the algorithm.
    We say a column (or row) is \textit{dirty} if it contains both elements with mark ``0'' and ``1''.
    After the first execution of steps 1 and 2, the number of dirty columns is reduced to at most  $w$ (Fig.~\ref{fig:dirt}).
    After the next execution of steps  3 and 4, the number of dirty rows is reduced to two.
    Crucially, the two dirty rows are adjacent.
    After one more execution of steps 1 and 2, there would be two dirty columns, however, since the rows were
    ordered in alternating order in step 1, one of the dirty rows will have ``0''s towards the top and ``1'' towards the
    bottom, while the other will them in reverse order.
    This means, after step 3, there will be only one dirty column and only one dirty row after step 4.
    The final sorting round will sort the elements in the dirty row and thus the elements will be correctly sorted
    by their mark.

    As the algorithm is oblivious to marks, the matrix will be sorted by marks regardless of 
    our choice of the marking value, meaning, the matrix will be correctly sorted.
    Conversions between column-major and row-major (and vice versa) can be performed in $O(m)$ time using~\cite{catanzaro:transpose}
    while respecting CAC.
    Sorting rows is the main runtime bottleneck and the only part that needs more than $O(m)$ time.
\end{proof}

\begin{corollary}
    The partition problem can be solved in $O(m)$ time if $w \le \sqrt{m}$.
\end{corollary}
\begin{proof}
    Observe that for the partition problem, we can ``sort'' each row in $O(m)$
    time (e.g., using radix sort).
\end{proof}

Using the above as a base case, we can now sort a square matrix efficiently.
\begin{theorem}\label{thm:square}
    If $w = m$, then the sorting problem can be solved in $O(m \log m)$ time. 
    The partition problem can be solved in $O(m)$ time.
\end{theorem}
\begin{proof}
    Once again, we use the idea of the 0-1 principle and assume the elements are marked with ``0'' and ``1'' bits,
    invisible to the algorithm.
    To sort an $m \times m$ matrix, partition the rows into groups of $\sqrt{m}$ adjacent rows. 
    Call each group a \textit{super-row}. 
    Sort each super-row using Lemma~\ref{lem:fat}. 
    Each super-row has at most one dirty row so there are at most $\sqrt{m}$ dirty rows in total.
    We now sort the columns of the matrix in ascending order by transposing the matrix, sorting the rows, and then 
    transposing it again.
    This will make all the dirty rows appear contiguously.
    We sort each super-row in alternating order (the first super-row ascending,
    next descending and so on).
    After this, there will be at most two dirty rows left, sorted in alternating order.
    We sort the columns in ascending order, which will reduce the number of dirty rows to one. 
    A final sorting of the rows will have the matrix in the correct sorted order.

    The running time is clearly $O(m\log m)$.
    In the partition problem, we use radix sort to sort each row and thus the running time will be $O(m)$.
\end{proof}

For non-square matrices, the situation is not so straightforward and thus  more interesting.
First, we observe that by using the $O(\log w)$ time EREW PRAM algorithm~\cite{cole:mergesort} to sort $w$ items using $w$ processors, the above algorithm can be generalized to non-square matrices:

\begin{theorem}
    A $w  \times m$ matrix $\M$, $w \ge m$, can be sorted in $O(m \log w)$ time.
\end{theorem}
\begin{proof}
    As before assume ``0'' and ``1'' marks have been placed on the elements in $\M$.
    First, we sort the rows. Then we sort the columns of $\M$ using the EREW PRAM algorithm~\cite{cole:mergesort}.
    Next, we convert the matrix from column-major to row-major layout, meaning, in the first column, the first $m$ elements 
    will form the first row, the next $m$ elements the next row and so on.
    This will leave at most $m$ dirty rows.
    Once again, we sort the columns, which will place the dirty rows contiguously.
    We sort each  $m \times m$  sub-matrix in alternating order, 
    using Theorem~\ref{thm:square}, which will reduce the number of dirty rows to two.
    The rest of the argument is very similar to the one presented for Theorem~\ref{thm:square}.
\end{proof}

The next result might not sound very strong, but it will be very useful in the next section.
It does, however, reveal some differences between the partition and sorting problems.
Furthermore, it is optimal as long as $w$ is polynomial in $m$.

\begin{lemma}\label{lem:par}
    The partition problem on a $w \times m$ matrix $\M$,  $ w \ge m > 2 \sqrt{\log w}$, 
    can be solved in $O(m \log^3_m w)$ time.
    The same bound holds for sorting $O(\log n)$-bit integers.
\end{lemma}
\begin{proof}
    We use the idea of 0-1 principle, combined with radix sort, which allows us to sort every row in $O(m)$ time. 
    The algorithm uses the following steps.
    \para{Balancing.}
    Balancing has $\lceil \log_m w \rceil$ rounds; for simplicity, we assume $\log_m w$ is an integer.
    In the zeroth round, we create $w/m$ sub-matrices of dimensions $m\times m$ by partitioning the rows into groups of $m$ contiguous rows
    and then sort each sub-matrix in column-major order (i.e., the elements marked ``0'' will be to the left 
    and the elements marked ``1'' to the right).
    At the beginning of each subsequent round $i$, we have $w/ m^i$ sub-matrices with dimensions $m^i \times m$.
    We partition the sub-matrices into groups of $m$ contiguous sub-matrices; 
    from every group, we create $m^i$ square $m\times m$ matrices:
    we pick the $j$-th row of every sub-matrix in the group, for $1 \le j \le m^i$, to create the $j$-th square matrix in the group.
    We sort each $m \times m$ matrix in column-major order. 
    Now, each group will form a $m^{i+1} \times m$ sub-matrix for the next round.
    Observe that balancing takes $O(m \log_m w)$ time in total.

    \para{Convert and Divide (C$\&$D).}
    With a moment of thought, it can be proven that after \textit{balancing}, each row will have between $\frac{w_0}{w} -\log_m w$ and $\frac{w_0}{w} + \log_m w$
    (resp. $\frac{w_1}{w} -\log_m w$ and $\frac{w_1}{w} + \log_m w$) elements
    marked ``0'' (resp. ``1'') where $w_0$ (resp. $w_1$) is the total number of elements marked ``0'' (resp. ``1'').
    This means, that there are at most $d=2\log_m w$ dirty columns.
    We convert the matrix from column-major to row-major layout, which leaves us with $\frac{w}{m d}$ contiguous dirty rows (each column is placed
    in $\frac{w}{m}$ rows after conversion). 
    Now, we divide $\M$ into $md$ smaller matrices of dimension $\frac{w}{md} \times m$. Each matrix will be a new subproblem.

    \para{The algorithm.}
    The algorithm repeatedly applies \textit{balancing} and \textit{C$\&$D} steps:
    after the $i$-th execution of these two steps, 
    we have $(md)^i$ subproblems where each subproblem is a  $\frac{w}{(md)^i} \times m$ matrix, and in the
    next round, \textit{balancing} and \textit{C$\&$D} operate locally within each subproblem (i.e., they work on $\frac{w}{(md)^i} \times m$ matrices).
    Note that before the divide step, each subproblem has at most 
    $\frac{w}{(md)^i}$ contiguous dirty rows.
    After the divide step, these dirty rows will be sent to at most two different sub-problems, meaning, at the depth $i$ 
    of the recursion, all the dirty rows will be clustered into at most $2^i$ groups of contiguous rows, each
    containing at most $\frac{w}{(md)^i}$ rows.
    Since, $m > 2 \sqrt{\log_m w}$, after  $2\log_m w$ steps of recursion, the subproblems will have at most $m$ rows and thus each subproblem can be sorted in
    $O(m)$ time, leaving at most one dirty row per sub-problem. 
    Thus, we will end up with only $2^{2\log_m w }$ dirty rows in the whole matrix $\M$.
    
    We now sort the columns of the matrix $\M$:
    we fit each column of $\M$ in a $\frac{w}{m} \times m$  submatrix and recurse on the submatrices using the
    above algorithm, then convert each submatrix  back to a column.
    We will end up with a matrix with $2^{2\log_m w }$ dirty rows but crucially,
    these rows will be contiguous. 
    Equally important is the observation that $ 2^{2\log_m w } \le m/2$ since $m > 2\sqrt{\log w}$.
    To sort these few dirty rows, we decompose the matrix into square matrices of $m \times m$, and sort each square matrix.
    Next, we shift this decomposition by $m/2$ rows and sort them again.
    In one of these decompositions, an $m \times m$ matrix  will fully contain all the dirty rows, meaning, we will end up with a fully sorted matrix.
    If $f(w,m)$ is the running time on a $w \times m$ matrix, then we have the recursion 
    $f(w,m) = f(w/m,m) + O(m \log^2_m w)$ which gives our claimed bound.
\end{proof}

\section{A Randomized Algorithm for Permutation}
In this section we present an improved algorithm for the permutation problem that  beats our 
best algorithm for partitioning for when the matrix is ``tall and narrow'' (i.e., $w \gg m$).
We note that the algorithm is only of theoretical interest as it is a bit technical and it uses
randomization. 
However, at least from this theoretical point of view, it hints  that perhaps in 
the GPU model, the permutation problem and the partitioning problem are different and require
different techniques to solve.

Remember that we have a matrix $\M$ which contains a set of elements
with labels. The labels form a permutation of $[w] \times [m]$, and an element with label $(i,j)$
needs to be placed at the $j$-th memory location of row $i$.
From now on, we use the term ``label'' to both refer to an element and its label.

To gain our speed up, we use two crucial properties: first, we use randomization 
and second we use the existence of the second index, $j$, in the labels.

\para{Intuition and Summary.}
Our algorithm works as follows. 
It starts with a preprocessing phase where
each row picks a random word and these random words are used to 
shuffle the labels into a more ``uniform'' distribution.
Then, the main body of the algorithm begins. 
One row picks a random hash function and communicates it to the rest of the rows.
Based on this hash function, 
all the rows compute a common coloring of the labels using $m$ colors $1, \ldots, m$ such that
the following property holds: for each index $i$, $1 \le i \le w$, and among the labels $(i,1),(i,2), \ldots, (i,m)$,
there is exactly one label with  color $j$, for every $1 \le j \le m$. 
After establishing such a coloring, in the upcoming $j$-th step, each row will 
send one label of color $j$ to its correct destination. 
Our coloring guarantees that CAC is not violated.
At the end of the $m$-th step, 
a significant majority of the labels are placed at their correct position and thus
we are left only with a few ``leftover'' labels. 
This process is repeated until the number of ``leftover'' labels is a  $1/\log_m^3 w$ fraction of the original input.
At this point, we use Lemma~\ref{lem:par} which we show can be done in $O(m)$ time since
very few labels are left.
Unfortunately, there will be many technical details to navigate.

\para{Preprocessing phase.}
For simplicity, we assume $w$ is divisible by $m$.
Each row picks random a number $r$ between 1 and $m$ and shifts its labels 
by that amount, i.e., places a label at position $j$ into position $j+r \inmod{m}$.
Next, we group the rows into matrices of size $m \times m$ which we transpose (i.e., the first $m$
rows create the first matrix and so on).

\para{The main body and coloring.}
We use a result by Pagh and Pagh~\cite{PP08}.
\begin{theorem}\cite{PP08}\label{thm:hash}
    Consider the universe $[w]$ and a set $S \subset [w]$ of size $m$.
    There exists an algorithm in the word RAM model 
    that (independent of $S$) can select a family $H$ of hash functions from $[w]$ to $[v]$ 
    in $O(\log m (\log v)^{O(1)})$ time and using $O(\log m + \log\log w)$ bits,
    such that:
    \begin{itemize}
        \item $H$ is $m$-wise independent on $S$ with probability $1- m^{-c}$ for a constant $c$
        \item Any member of $H$ can be represented by a data structure that uses 
            $O( m \log v )$ bits and the hash values can be computed in constant time. 
            The construction time of the structure is $O(m)$.
    \end{itemize}
\end{theorem}
The first row builds $H$ then picks a hash function $h$ from the family. 
This function, which is represented as a data structure, is communicated to the rest of the rows
in $O(\log w + m)$ time as follows: 
assume the data structure consumes $S_m = O(m)$ space.
In the first $S_m$ steps, the first row sends the $j$-th word in the data structure to row $j$. 
In the subsequent $\log(w)$ rounds, 
the $j$-th word is communicated to any row $j'$
with $j' \equiv j \inmod{S_m}$, using a simple broadcasting strategy that doubles the number
of rows that have received the $j$-th word.
This boils down the problem of transferring the data structure to the problem of distributing
$S_m$ random words between $S_m$ rows which can be easily done in $S_m$ steps while respecting CAC.

\para{Coloring.}
A label $(i,j)$ is assigned color $k$ where 
$j \equiv h(i) + k\inmod{m}$.

\para{Communication.}
Each row computes the color of its labels. 
Consider a row $i$ containing some labels. 
The communication phase has $\alpha m$ steps, where $\alpha$ is a large enough constant.  During step $k$, $1 \le k \le m$,
the row $i$ picks a label of color $k$. If no such label exists, then the row
does nothing so let's assume a label $(i_k, j_k)$ has assigned color $k$. 
The row $i$ sends this label to the destination row $i_k$.
After performing these initial $m$ steps, the algorithm repeats these $m$ steps $\alpha-1$ times
for a total of $\alpha m$ steps. 
We claim the communication step respects CAC; assume otherwise that 
another label $(i'_k,j'_k)$ is sent to row $i_k$ during step $k$.
Clearly, we must have $i'_k = i_k$ but this contradicts our coloring since it implies
$j'_k \equiv h(i_k)+ k \equiv j_k \inmod{m}$ and thus $j'_k = j_k$.
Note that the communication phase takes $O(m)$ time as $\alpha$ is a constant.

\para{Synchronization.}
Each row computes the number of labels that still need to be sent to their destination,
and in $O(\log w + m)$ time, these numbers are summed and broadcast to all the rows. 
We repeat the main body of the algorithm as long as 
this number is larger  than $\frac{wm}{\log^3_m} w$.

\para{Finishing.} 
We leave the details here for the appendix but roughly speaking, we do the following:
first, we show that we can pack the remaining elements into a 
$w \times O(\frac{m}{\log^3_m w})$ matrix in $O(m)$ time (this is not trivial as the matrix can have
rows with $\Omega(m)$ labels left).
Then, we use Lemma~\ref{lem:par} to sort the matrix in $O(m)$ time. 
Finishing off the sorted matrix turns out to be relatively easy.

\para{Correctness and Running Time.}
Correctness is trivial as the end all the rows receive their labels and the algorithm is shown to respect CAC.
Thus, the main issue is the running time.
The \textit{finishing} and \textit{preprocessing} steps clearly takes $O(m)$ time.
The running time of each repetition of \textit{coloring}, \textit{communication} and \textit{synchronization} is $O(\log w + m)$
and thus to bound the total running time we simply need to bound how many times the synchronization steps is executed.
To do this, we need the follow lemma (proof is in the appendix).

\begin{lemma}\label{lem:color}
    Consider a row $\ell$ containing a set $\s_\ell$ of $k$ labels, $k \le m$.
    Assume the hash function $h$ is $m$-wise independent on $\s_\ell$. 
    For a large enough constant $\alpha$, let $\C_{\hev}$  be the set of
    colors that appear more than $\alpha$ times in $\ell$, 
    and $\s_{\hev}$ be the set of labels assigned colors from $\C_{\hev}$.
    Then $\E[|S_{\hev}|]=m (k / 2m)^\alpha$ and
    the probability that $|S_{\hev}| = \Omega(m (k/2m)^{\alpha/10})$
    is at most $2^{-\sqrt{m} (k / 2m)^{\alpha/10}}$.
\end{lemma}

Consider a row $\ell$ and let $k_i$ be the number of labels in the row after the $i$-th iteration of the \textit{synchronization}
($k_0 = m$).
We call $\ell$ is a \textit{lucky} row if the hash function is $m$-wise independent on $\ell$ during all the $i$ iterations.
Assume  $\ell$ is lucky.
During the \textit{communication} step, we process all the labels except those in the set $\s_{\hev}$.
By the above lemma, it follows that with high probability, we will have $O( m (k_i/2m)^{a/10})$ labels left 
for the next iteration on $\ell$, meaning, with high probability, $k_{i+1} = O( m (k_i/2m)^{a/10})$.

Observe that if $m \ge \log^{O(1)} w$ (for an appropriate constant in the exponent) and $k_i > m / \log_m^3 w$, then 
$\sqrt{m} \left( \frac{k_i}{2m} \right)^{\alpha/10} > 3 \log w$ which means the probability that Lemma~\ref{lem:color} ``fails''
(i.e., the ``high probability'' fails) is at most $1/w^3$. 
Thus, with probability at least $1-1/w$, Lemma~\ref{lem:color} never fails for any lucky row and thus each lucky
row will have at most 
\[ O\left( m \left( \frac{k_0}{2m} \right)^{(\alpha/10)^i} \right)\]
labels left after the $i$-th iteration. 
By Theorem~\ref{thm:hash}, the expected number of unlucky rows is at most $n/m^c$ at each iteration
which means  after $i=O(\log\log\log_m w)$ iterations, there will be $O( \frac{mw}{ \log_m^3 w} + \frac{mw i}  {m^c })= O( {mw} {\log_m^3 w})$ labels left.
So we proved that the algorithm repeats the \textit{synchronization} step at most $O(\log\log\log_m n)$ times, giving us the following theorem.

\begin{theorem}\label{thm:par}
    The permutation problem on an $w \times m$ matrix can be solved in $O(m \log\log\log_m w)$ expected
    time, assuming $m = \Omega(\log^{O(1)}w)$.
    Furthermore, the total number of random words used by the algorithm is $w + O(m \log\log\log_m w)$.
\end{theorem}



\bibliographystyle{splncs03}
\bibliography{refs}

\newpage
\appendix 
\section{Proof of Lemma~\ref{lem:color}}
    We partition the set of $k$ labels in the row $\ell$ into equivalent classes $\s_1, \cdots, \s_w$,
    where $\s_i$ includes all the labels with row destination $i$ (obviously, most of these sets will be
    empty, since $w \gg m \ge k$).

    We claim that after the preprocessing step, $|\s_i| = O(\log w)$ with high probability.
    Remember that row $\ell$ was part of a  $m\times m$ matrix $\M'$ during the preprocessing
    phase. Also observe that due to the way the preprocessing step is set up, row $\ell$ will
    receive a uniform random label from each row of $\M'$ and crucially, these random
    labels are independent. 
    Now, if we consider $\s_i$, it follows that $|\s_i|$ is the sum of $m$ independent 0-1 random
    variables (but not identical) with constant mean and thus our claim follows after a 
    relatively straightforward application of Chernoff inequality.
    Also observe that we have \[ \sum_{i=1}^w |\s_i| = k. \]

    Now, consider a fixed color $j$. 
    The probability that a label in $\s_i$ receives color $j$ is $|\s_i| / m$; this
    follows since we have assumed that the hash function $h$ does not fail on $\ell$ which means 
    the colors received by labels in $\s_i$ are dependent and distinct 
    (i.e., knowing the color of one label in $\s_i$ forces the remaining colors) but they are
    independent of colors received by all the other equivalent classes.
    Let $E_{t,j}$ be the event that $t$ labels are assigned color $j$ in this row.
    Due to $m$-way independence of colors from different equivalent classes, 
    \[ \Pr[E_{t,j}] = \sum_{\substack{I\subset [w] \\ |I| = t}}  \prod_{i\in I} \frac{|\s_i|}{m} < \frac{1}{t!} \left( \sum_{i=1}^w \frac{|\s_i|}{m}  \right)^t  = \frac{1}{t!}\left( \frac{k}{m} \right)^t.  \]
    Using Stirling's formula, we can show that there exists a constant $\alpha$ such that 
    $\E[|\s_i|] = \sum_{t=\alpha}^\infty t \Pr[E_{t,j}] \le \left( \frac{k}{2m} \right)^\alpha$ 
    which proves the first part of the lemma.

    Proving the high concentration bound is more involved. 
    We now conceptually partition the colors into $\sqrt{m}$ \textit{color-subsets}, 
    $\C_1, \cdots, \C_{\sqrt{m}}$ in the following way: for each color we independently and uniformly pick its 
    color-subset (with probability $\frac{1}{\sqrt{m}}$).
    Consider a particular equivalent class $\s_i$. 
    As discussed, there are only $m$ possible different choices for colors to be assigned to labels of $\s_i$
    and in particular, knowing the color of any label in $\s_i$ determines the remaining colors.
    Consider one of these $m$ possible color assignments to labels of $\s_i$. 
    For every $j$, $1 \le j \le {\sqrt{m}}$, we consider the ``bad'' event when more than 10 colors from color-subset
    $\C_j$ have been assigned to labels of $\s_i$. 
    Since the probability of including any color in $\C_j$ is $\frac{1}{\sqrt{m}}$, 
    it follows that the probability of this bad event is at most $O\left( |\s_i|^{10} \sqrt{m}^{-10}\right) = O\left( |\s_i|^{10} m^{-5}\right)$.
    Remember that for every $\s_i$ there are $m$ such bad events (one for every color assignment), 
    the number of non-empty equivalent classes $\s_i$ is at most $m$,
    $\s_i$ has $O(\log w)$ labels, and $m$ is at least polylogarithmic in $w$; using all these facts, we can deduce that with some positive probability, none of the
    bad events happen.
    Thus, there exists a partition of the set of colors into $\sqrt{m}$ color-subsets such that no bad event happens. 
    Using another straightforward application of Chernoff bound, we can also guarantee that $|\C_j| = \Theta(\sqrt{m})$.
    For the rest of this proof, assume our partition into color-subsets is one of them (remember that this 
    partitioning is only conceptual and it is not used during the algorithm).

    We look at number of labels assigned colors from $\C_j \cap \C_{\hev}$, for a fixed index $1 \le j \le \sqrt{m}$.
    If we can show high concentration bound for this number, then the second part of our lemma follows (using union bound).
    To do this, for every subset $X \subset \C_j$, we estimate the probability that at least $\alpha |X|$
    labels receive colors from $X$.
    For $|X| =  \sqrt{m} \left(c \frac{k}{m} \right)^{\alpha/10}$ and large enough constants $c$ and $\alpha$, 
    this is bounded by 
    \begin{align*} 
        &\binom{|\C_j|}{|X|} \sum_{I= \alpha|X|/10}^m \binom{k}{I} \left(   \frac{X}{m}\right)^{I}  \le  \left( \frac{|\C_i|}{|X|} \right)^{|X|} \sum_{I= \alpha|X|/10}^m \left( \frac{k}{I} \right)^{I} \left( \frac{X}{m} \right)^{I} \ll \\ 
        & \left( \frac{|\C_i|}{|X|} \right)^{|X|} \left( \frac{k |X|}{m  \alpha|X|/10} \right)^{\alpha |X|/10} \le \Theta\left( \frac{m}{ck}  \right)^{\alpha|X|/10} \left( \frac{10k}{m  \alpha} \right)^{\alpha |X|/10 }.
    \end{align*}
    where the notation $\ll$ is used to replace the O-notation, i.e., $f \ll g$ denotes $f = O(g)$.
    The same argument can be applied to the remaining $\sqrt{m}$ color-subsets and thus the last part of the
    lemma follows from an easy application of union bound.

\section{Finishing off in  Theorem~\ref{thm:par}}\label{sec:ap:pack}
Consider the matrix $\M$ after the last \textit{synchronization} step; there are less than $ \frac{wm}{\log^3_m w}$ labels left.
As discussed, first we like to pack the remaining labels into a $w \times O(\frac{m}{ \log^3_m w})$ matrix in $O(m)$ time. 
This is not trivial as the matrix can have rows with $\Omega(m)$ labels left.

Fortunately, we will use the fact that the majority of the rows will have
$O(m / \log^3_m w)$ labels left.
To see, this, consider the hash function $h$ used in the algorithm. 
The probability that the hash function fails for any row  is at most $1 / m^c$ (where $c$ is the constant defined in Theorem~\ref{thm:hash}).
We call a row where this hash function never fails a \textit{lucky} row and otherwise we call it
an \textit{unlucky} row.
We now claim that all the lucky rows are ``narrow'' and they have $O(m / \log^3_m n)$ labels.
Observe that if $m \ge \log^{O(1)} w$ (for an appropriate constant in the exponent) and $k_i > m / \log_m^3 w$, then 
$\sqrt{m} \left( \frac{k_i}{2m} \right)^{\alpha/10} > 3 \log w$. 
Thus, for every lucky row, the ``high probability'' bound in  Lemma~\ref{lem:color} is at least $1-1/w^3$,
which means, we can assume Lemma~\ref{lem:color} never fails (the expected increase in the running time for when the
lemma fails is negligible, since we can afford to spend $O(w^2)$ time in such rare cases).
As the analysis shows, we can prove that after $r=O(\log\log\log_m n)$ iterations of the \textit{synchronization} step, 
the number of labels left in each row is $O(m / \log^3_m n)$.
It thus follows that the expected number of unlucky rows is at most $O( rw / m^{c}) < w/m^{c-1}$.
By Markov inequality, the probability that the number of unlucky rows is more than $w/m^{c-3}$ is at most 
$1/m^2$ which means even if the algorithm runs in $O(m^2)$ time, then the increase in the expected
running time is again negligible.
Thus, we assume the number of unlucky rows is at most $w/m^{c-3}$.

To pack the matrix $\M$,  we pick another collection of $t=\Theta(\log^2 w)$ random integers
$t_0, \cdots, t_{t-1}$, between 1 and $w$ and broadcast them to all the rows.
Next, for each $i$, $1 \le i \le t$, row $j$ communicates with row $j' \equiv j+t_i \inmod{w}$ and these rows transfer 
up to $O(m/t)$ labels from the row containing the most labels to the row containing the least labels, as follows:
if both of these rows are lucky or both unlucky, then they do nothing. 
So assume, $j'$ is lucky and $j$ is unlucky. 
If $j$ is the first unlucky row matched to $j'$, then $j$ transfers $2m/t$ labels to $j'$, otherwise, nothing happens.
Now, lets fix a particular unlucky row $j$.
The probability $j$ matches to an unlucky row is very small, at most $1/m^{c-3}$.
Crucially, since the random numbers $t_1, \cdots, t_{t-1}$ are independence, the rows 
$j$ matches to are also independent. 
This means, by Chernoff inequality, with high probability, $j$ will match to at least
$t/2$ lucky rows.
Each such match will relieve $j$ of $2m/t$ labels, meaning, $j$ can be relieved of all of its labels at the end.
Each matching requires $O(m/t)$ time and thus the total time is $O(m)$.

Now we have a packed matrix $\M$.

We use Lemma~\ref{lem:par} to sort the packed matrix.
After sorting, consider a row containing labels $(i_1, j_1), \cdots, (i_t, j_t)$, sorted lexicographically, that is
either $i_{k} < i_{k+1}$ or if $i_{k} = i_{k+1}$ then $j_k < j_{k+1}$.
We have three types of labels $(i_k,j_k)$ here: one, labels  where $i_k = i_1$ 
(there are at the beginning of the list), two, labels where $i_k = i_t$ (these are at the end of the list), and three the rest (everything else that lies in the middle).
The labels in the third group have the property that no other row contains another label
$(i'_k,j'_k)$ with $i_k = i'_k$ (since the labels were sorted) which means, row $i$ can safely send 
any label $(i_k,j_k)$ to row $i_k$ without violating CAC. 
Thus, in $m$ time, we transfer all the middle labels.
Next, we send $(i_k,j_k)$ (with $i_k = i_1$) to row $i_k$ during step the next $j_k$-th step which
takes care of the labels in the first group.
The labels in the second group can be handled similarly.

\end{document}